\documentclass[letterpaper, 10 pt, conference]{ieeeconf}  % Comment this line out if you need a4paper

\IEEEoverridecommandlockouts                              % This command is only needed if 
\usepackage{cite}

\usepackage{subfig}

\usepackage{amsmath,amssymb,amsfonts}
\usepackage{amsthm}
\usepackage{bbold}
\usepackage{graphicx}
\usepackage{textcomp}
\usepackage{xcolor,subfig,multirow,caption}

\newtheorem{theo}{Theorem}

\newtheorem{exmp}{Example}[section]
\usepackage[norelsize, linesnumbered, ruled, lined, boxed, commentsnumbered, noend]{algorithm2e}

\usepackage[normalem]{ulem}
\usepackage{mathabx}

\ifodd 1
\newcommand{\com}[1]{{\color{red}\textbf{Comment}:#1}}
\else

\newcommand{\com}[1]{}
\fi

\newcommand{\comment}[1]{}

\title{\LARGE \bf
Group Formation through Game Theory and Agent-Based Modeling: Spatial Cohesion, Heterogeneity, and Resource Pooling
}

\author{Chenlan Wang$^{1}$
, Jimin Han$^{2}$% <-this % stops a space
, Diana Jue-Rajasingh$^{2}$ 
\thanks{$^{1}$ University of Michigan, Ann Arbor, USA. 
$^{2}$ Rice University, Houston, USA.
}
}

\begin{document}

\maketitle
\thispagestyle{empty}
\pagestyle{empty}
%%%%%%%%%%%% forcing page %%%%%%%%%%
% \thispagestyle{plain}
% \pagestyle{plain}

%%%%%%%%%%%%%%%%%%%%%%%%%%%%%%%%%%%%%%%%%%%%%%%%%%%%%%%%%%%%%%%%%%%%%%%%%%%%%%%%
\begin{abstract}
 This paper develops a game-theoretic model and an agent-based model to study group formation driven by resource pooling, spatial cohesion, and heterogeneity. We focus on cross-sector partnerships (CSPs) involving public, private, and nonprofit organizations, each contributing distinct resources. Group formation occurs as agents strategically optimize their choices in response to others within a competitive setting. We prove the existence of stable group equilibria and simulate formation dynamics under varying spatial and resource conditions. The results show that limited individual resources lead to groups that form mainly among nearby actors, while abundant resources allow groups to move across larger distances. Increased resource heterogeneity and spatial proximity promote the formation of larger and more diverse groups. These findings reveal key trade-offs shaping group size and composition, guiding strategies for effective cross-sector collaborations and multi-agent systems.
\end{abstract}

%
% \input{intro}
%%%%%%%%%%%%%%%%%%%%%%%%%%%%%%%%%%%%%%%%%%%%%%%%%%%%%%%%%%%%%%%%%%%%%%%%%%%%%%%%
\section{Introduction}
Group formation is a common phenomenon in both natural and engineered systems, often driven by the pursuit of collective benefits such as resource pooling, enhanced accuracy, and greater energy efficiency. For example, firefighters and robots collaborate to enable faster and more effective disaster response \cite{kumar2004robot, roldan2021survey}. Surgeons and surgical robots coordinate to improve precision during operations \cite{abbasi2024integration}. Concertgoers often use ridesharing to lower travel costs. Birds migrate in flocks to improve foraging efficiency and reduce predation risk \cite{alerstam2011optimal}. Similarly, cross-sector partnerships (CSPs) unite government, business, and nonprofit organizations to address complex social problems through resource sharing, complementary expertise, and coordinated action \cite{quelin2017public, selsky2005cross}.

However, group formation is not without costs. These include maintenance, communication overhead, coordination delays, and travel constraints. For example, there may be a maximum operational range for a firefighter or operator to control a rescue robot or drone effectively. Communication—whether between humans or robots—often requires ongoing synchronization, increasing energy demands, and system complexity \cite{klavins2004communication}. In animal groups, cohesion requires continuous adjustment and alignment, which can be energetically expensive \cite{alerstam2011optimal}. In CSPs, reconciling various goals, values, and decision-making processes can lead to negotiation burdens, trust-building challenges, and delay in action \cite{ashraf2017animosity}. In human organizations more broadly, group decision-making can slow down responses due to conflicting preferences or hierarchical structures. In addition, larger groups may suffer from free-riding problems, reduced individual accountability, or an increased risk of resource depletion.

The formation of stable groups among strategic agents has been examined in various disciplines, including computer networks, economics, social science, public policy, political science, and biology
\cite{couzin2006behavioral,  neal2016compatibility, hollard2000existence, milchtaich2002stability}.
A wide range of models have been developed to explore this process, such as coalition formation games \cite{hagen2020two, hoefer2018dynamics, wang2024structural}, 
sequential games \cite{wang2024stackelberg, wang2023cooperation, wang2025decision},
clustering models \cite{van2013community}
% , wittemyer2005socioecology}, 
and agent-based modeling approaches
\cite{collins2017agent, collins2018strategic}.

This work asks: What are the tradeoffs between spatial cohesion, heterogeneity, and resource pooling in shaping emerging stable group structures? In particular, how do these factors affect group size and composition? We frame this as an optimization problem in a strategic setting.

We analyze this group formation problem through a game-theoretic lens, where equilibrium reflects that each agent has reached a locally optimal decision. Our study builds upon the Spatial Group Formation Game introduced by \cite{wang2024structural}, which analyzed the trade-off between spatial cohesion and resource pooling among homogeneous individuals—i.e., agents of a single affiliation type. While that model offered valuable insights into how agents self-organize spatially to share resources, it did not account for heterogeneity among agents. 

In many real-world scenarios, especially in the formation of cross sector partnerships (CSPs), group members differ not only in the quantity of resources they possess but also in their organizational type, such as public, private, or nonprofit. Studying group formation in such heterogeneous settings is essential for understanding how diverse actors coordinate to achieve shared goals despite structural and strategic differences. This heterogeneity often takes the form of resource complementarity, a key driver of CSPs, where organizations from different sectors contribute distinct but interdependent assets that, when combined, enhance collective problem-solving capacity \cite{selsky2005cross, mahoney2009perspective}. The resource-based view (RBV) suggests that organizations form partnerships to access capabilities they lack, leveraging each other’s strengths to pursue common objectives \cite{barney1991firm}. For example, private firms may provide financial capital and technological expertise, while public institutions contribute legitimacy, regulatory support, and access to large scale infrastructure \cite{mahoney2009perspective}. These complementary resources can improve partnership outcomes by fostering innovation, increasing operational efficiency, and enhancing the overall impact on societal challenges \cite{selsky2005cross}. However, while resource diversity enables synergy, it also introduces complications such as coordination difficulties, institutional misalignment, and governance complexity. These challenges require strategic management to ensure partnership effectiveness \cite{ashraf2017animosity}.

The formation of CSP, in particular, can be viewed as a problem of group formation in strategic interactions, where organizations decide to join or form coalitions based on mutual benefit, while anticipating the choices of others. Stability depends on whether any member has an incentive to leave for a better option. Motivated by this, we extend the previous model to a one-shot spatial group formation game with heterogeneous agents, allowing us to explore how organizational diversity shapes coalition dynamics in CSPs. Beyond CSPs, our model can be broadly applied to other domains involving heterogeneous multi-agent interactions, such as human–robot teaming, joint disaster response, collaborative sensing, and resource-sharing in smart infrastructure systems.

We begin by introducing our game model as a simultaneous group formation game in Section \ref{game model}, followed by a proof of equilibrium existence in Section \ref{equlibrium}. In Section \ref{ABM}, we employ agent-based modeling to visualize the dynamics of the group formation process and examine the resulting group structures at equilibrium. Section \ref{sec:discussion} presents numerical results from the agent-based simulations, highlighting how spatial cohesion, heterogeneity, and resource pooling influence the size and composition of stable groups. Finally, we conclude the paper in Section \ref{sec:conclusion}.

%%%%%%%%%%%%%%%%%%%%%%%%%%%%%%%%%%%%%%%%%%%%%%%%%%%%%%%%%%%%%%%%%%%%%%%%
\section{Model and Preliminaries}
\subsection{The Game-Theoretic Model}\label{game model}

Our model builds upon the Spatial Group Formation Game Model introduced by \cite{wang2024structural}, which examines the tradeoff between spatial cohesion and resource pooling among individuals of a single type. We extend the model by incorporating heterogeneous individuals with distinct types and resource allocations to account for heterogeneity. In our work, we use the terms ``individuals", ``agents", ``actors", and ``organizations" synonymously. Similarly, ``group" and ``cross-sector partnership" are used interchangeably.

In our model, there are $n \geq 2$ individuals or organizations,  each indexed by $i \in \mathcal{N}=\{1,2,...,n\}$. Each organization belongs to one of three categories or types: public sector, private sector, or non-profit sector, denoted as $t_i \in \{c_1, c_2, c_3\} $. More generally, the model allows for $k \geq2$ categories, so that $t_i \in \{c_1, c_2, \ldots, c_k\}$. In this paper, we focus on the context of cross sector partnerships, where $k=3$, corresponding to the public, private, and nonprofit sectors. Each organization is characterized by a physical location $\mathbf{x_i} \in \mathbb{R}^2$, and a resource vector $r_i = (r_i^{c_1}, r_i^{c_2}, r_i^{c_3})$.The term ``resource" is used in a broad sense to represent inherent attributes such as capabilities, skills, or sector-specific characteristics, enabling differentiation among sector types. The allocation of resources is type-specific and follows the constraints:
\begin{itemize}
    \item if $t_i =$ public, then $r_i = (r_i^{c_1},0,0)$ with $r_i^{c_1}>0$;
    \item if $t_i =$ private, then $r_i = (0,r_i^{c_2},0)$ with $r_i^{c_2}>0$;
    \item if $t_i =$ non-profit, then $r_i = (0,0,r_i^{c_3})$ with $r_i^{c_3}>0$.
\end{itemize}
This formulation ensures that each organization possesses resources exclusively associated with its designated category. Although resources may not be strictly exclusive across different sectors in practice, we impose this assumption as a modeling simplification. 

We study how cross-sector partnerships are formed as a one-shot game. Each individual organization chooses the other organizations with whom it wants to be in the same CSP. We apply the same definitions for preferences and preference profiles from the game model by \cite{wang2024structural}: organization $i$'s strategy/action/preference space is $\mathcal{A}_i = \{ A \subseteq \mathcal{N}|i \in A\}$; sector $i$'s action by $a_i \in \mathcal{A}_i$; and a joint action profile is denoted by $\mathbf{a} = (a_1, a_2, \ldots, a_n)$.

Next, we assume that each organization derives equal benefits from its cross-sector partnerships, meaning that the utilities of agents within the same group are identical (i.e., $u_{i:i\in G} = U_{G}$). Although the distribution of group benefits may vary based on individual contributions, such as individual resources, this assumption serves as a simplification. Additionally, we disregard impacts such as competition among different cross-sector partnerships.

The utility function for each group is defined as:
\begin{equation}\label{eq:utility}
    U_G(\mathbf{t}, \mathbf{r}, \mathbf{x}) = f(R_G,D_G)
\end{equation}
where $\mathbf{t}=\{t_i: i\in G\}$, $\mathbf{r}=\{r_i: i\in G\}$, and $\mathbf{x}=\{\mathbf{x_i}: i\in G\}$. The function $f(R_G, D_G)$ is monotonically increasing with respect to $R_G$ and monotonically decreasing with respect to $D_G$. For instance, $f(R_G, D_G) = \frac{R_G}{e^{D_G}}$. In this context, $R_G$ and $D_G$ represent the aggregated resources and spatial coverage for group G, respectively. We call $R_G$ group resource, which captures the aggregated resources based on individual resources and their diversity in terms of types. Specially, the form of group resource is chosen as the following: 
\begin{equation}
    R_G = (R_G^{c_1} + R_G^{c_2} + R_G^{c_3})h(R_G^{c_1}, R_G^{c_2}, R_G^{c_3})
\end{equation}
Where $R_G^{c_1}$ represents the total resources in the first category (i.e. the resource of the public sectors in the CSP), defined as the sum of individual resources in the first category: $R_G^{c_1} = \sum r_i^{c_1}$. Similarly, $R_G^{pri}$ denotes the total resources in the second category (i.e. the resource of the private sectors in the CSP), calculated as $R_G^{c_2} = \sum r_i^{c_2}$, and $R_G^{c_3}$ signifies the total resources of the third category (i.e. the resource of the non-profit sectors in the CSP), given by $R_G^{c_3} = \sum r_i^{c_3}$. Therefore, $R_G^{c_1} + R_G^{c_2} + R_G^{c_3}$ represents the total resources of all individuals in the group.  The function $h$, referred to as the boosting factor due to heterogeneity, quantifies the enhancement in collective resources resulting from the presence of multiple resource categories within a group. In other words, it captures diversity based on the resources available in different sectors.
\begin{equation}
    h(R_G^{c_1}, R_G^{c_2}, R_G^{c_3}) = 1 + \frac{R_G^{c_1}R_G^{c_2} + R_G^{c_1}R_G^{c_3} + R_G^{c_2} R_G^{c_3} }{R_G^{c_1}+ R_G^{c_2}+R_G^{c_3}}
\end{equation}
Consequently, $R_G$ emphasizes the importance of heterogeneity in resource pooling. By integrating various resource types, groups can achieve collective outcomes that exceed the simple sum of individual resources. This synergy arises from the complementary nature of diverse resources, leading to enhanced efficiency and effectiveness.

\begin{exmp}\label{eg:RG}
Consider three cross-sector partnerships (CSPs) that possess the same total resources, defined as $R_G^{total} = R_G^{c_1}+ R_G^{c_2}+R_G^{c_3}$. For clarity, we introduce the collective resource vector $\mathbf{R_G} = (R_G^{c_1}, R_G^{c_2}, R_G^{c_3})$. The resource distributions for each CSP are as follows:
\begin{itemize}
    \item CSP1: $\mathbf{R_G} = (6,0,0)$, where all resources come from the public sector.
    \item CSP2: $\mathbf{R_G} =(4,2,0)$, with resources from both public and private sectors.
    \item CSP3: $\mathbf{R_G} = (2,2,2)$, where resources are distributed across public, private, and nonprofit sectors.
\end{itemize}
Although all three cross-sector partnerships (CSPs) have the same total resources ($R_G^{total} = 6$), we expect the effective collective resources to differ due to variations in sectoral diversity. Specifically, we anticipate the following relationship: $R_{CSP1} < R_{CSP2} < R_{CSP3}$. This is because CSP3 exhibits the highest diversity, involving all three sectors (public, private, and nonprofit), followed by CSP2 (public and private), while CSP1 lacks any diversity, relying solely on public sector resources. To quantify this effect, we calculate a diversity-enhanced factor $h$ for each CSP:
\begin{itemize}
    \item CSP1: $h_{CSP1} = 1 + 0/6 = 1$
    \item CSP2:  $h_{CSP2} = 1 + 4*2/6 = 7/3$
    \item CSP3:  $h_{CSP3} = 1 + (2*2+2*2+2*2)/6 = 3$
\end{itemize}
Using these values, the effective collective resources for each CSP are calculated as follows: $R_{CSP1} = 6$, $R_{CSP2} = 14$, $R_{CSP3} = 18$. These results confirm that greater sectoral diversity leads to higher effective collective resources.
\end{exmp}

$D_G$ is defined as the total pairwise distance between all sectors within the cross-sector partnership:
\begin{equation}
    D_G = \sum_{i,j\in G} d_{i,j} = \sum_{i,j\in G} ||\mathbf{x_i} - \mathbf{x_j}||
\end{equation}
where $d_{i,j}$ represents the Euclidean distance between sectors i and j. It is an effective and widely used measure of spatial cohesion for the group, as it quantifies the overall spread or dispersion of sectors within the partnership.

\subsection{Equilibrium and Existence}\label{equlibrium}
Our study aims to address the following key questions:
\begin{itemize}
\item How do agents, such as individual organizations from different sectors, form cross-sector partnerships (CSPs) while accounting for the interplay among resource pooling, heterogeneity, and spatial cohesion?
\item What dynamics emerge when agents continuously update their preferences in response to their environment or to the actions of others?
\end{itemize}

To explore these questions, we first examine whether stable groups, such as CSPs, can form under these strategic conditions. This involves analyzing the existence of an equilibrium in the one shot group formation game.

We adopt the concept of Individually Stable Equilibrium (ISE) as the solution concept. A partition of the set of individuals or organizations into one or more groups is considered individually stable if no agent can improve its utility by moving to another group that would accept its membership. A more formal definition is provided in \cite{wang2024structural}.

\begin{theo} \label{theo:best-response converge}
	There always exist an Individually Stable Equilibrium in the group formation game.
\end{theo}

\begin{proof}
   Algorithm 1, as introduced in \cite{wang2024structural}, is a general (asynchronous) improvement-update algorithm that can be directly applied to our game. Due to space limitations, we do not reproduce the algorithm here, but refer readers to \cite{wang2024structural} for full details. The central element of this algorithm is a potential function defined as the vector $\phi = [u_{[1]}, u_{[2]}, \ldots, u_{[n]}]$, which represents the individual utilities sorted in descending order.

Each update step in the algorithm ensures that this potential function does not decrease. Because there are only finitely many possible partitions of the player set, the number of distinct utility vectors is also finite. Therefore, the algorithm must eventually converge to a partition in which no player has an incentive to deviate unilaterally.

This convergence implies the existence of an individually stable equilibrium (ISE). That is, for any given set of organizations, there exists a partition into one or more groups such that no individual has an incentive to leave their group for another.
\end{proof}

To gain a deeper understanding of how cross-sector partnerships are formed dynamically, we extend our one-shot game into an agent-based model. This allows us to simulate the interactions between individual organizations over time, capturing the evolving preferences and decisions of each agent. By doing so, we can better analyze the dynamic processes underlying the formation of partnerships and the factors influencing their stability and evolution. The details of this approach are provided in the next section.

\subsection{The Agent-based Model}\label{ABM}
In the previous section, we showed that ISE always exists. Algorithm 1, described in \cite{wang2024structural}, can be directly applied to set up the agent-based model for the dynamics and equilibrium analysis of our one-shot game. 
\subsubsection{Initialization}
\begin{itemize} 
    \item Assign $m$ agents/organizations to each category or type, i.e., $m$ private organizations, $m$ public organizations, and $m$ nonprofit organizations. In total, there are $n=3m$ organizations. Note that selecting an equal number of organizations in each category is for illustrative and simplification purposes; the approach can be extended to arbitrary combinations.
    \item Assign each agent a 2D location by randomly selecting coordinates within the range $[0, x_{max}] \times [0, y_{max}]$.
    \item Assign each organization a category-specific resource by uniformly randomly selecting a value from the range $[1, r_{max}]$, where $r_{max} \geq 1$.
    \item Initialize the system with each group as a singleton, meaning each agent starts independently. In the context of cross sector partnerships (CSPs), a valid CSP must include organizations from at least two different sectors. Therefore, no true CSPs exist at the beginning of the process.
    \item Initialize the remaining agent set as the set of all agents/organizations. This set is used to track organizations that may have opportunities for improvement.
\end{itemize}
\subsubsection{Update} 
\begin{enumerate}
    \item Randomly select one agent/organization from the remaining agent set to update its choice by improving its utility. There are only two possible choices for the agent:
    \begin{itemize}
        \item Joining an existing group, provided the group accepts it (i.e., joining requires group permission and will not harm the utility of the group).
        \item Setting up its own singleton group.
    \end{itemize}
    \item Update the remaining agent set:
    \begin{itemize}
        \item If neither of the two choices improves the utility of the agent, no update will be made by the agent. In this case, remove the agent from the remaining agent set.
        \item If an update occurs from the agent, reset the remaining agent set to include all agents.
    \end{itemize}
    \item Repeat the first step.
\end{enumerate}
\subsubsection{Convergence}
If no agent can improve their utility by joining a group that accepts them, the groups are considered stable, referred to as an Individually Stable Equilibrium. This stability is guaranteed by the proof for Algorithm 1.

\section{Numerical Results and Discussion}\label{sec:discussion}
We use Python 3.8.19 to simulate the agent-based model described above. Initially, we visualize the dynamics of organization updates throughout the process until stable groups are formed. The primary focus of our results is to analyze the patterns of stable groups under different initialization settings.

In the settings, we select $m = 5$ organizations per category, resulting in a total of 15 individuals. For simplicity, we set $x_{max} = y_{max}$ and select the list $\{1, 10, 20, 40, 60, 80, 100\}$ for $x_{max}$. For the upper bound of individual resources, we select the list $\{1, 20, 40, 60, 80, 100\}$ for $r_{max}$. For each setting, we run 1000 simulations with random initialization based on the chosen parameters. In other words, the results presented reflect the average patterns observed across these 1000 games. For each game, it reaches stable group formation in fewer than 200 iterations.

\begin{figure}[ht]
    \centering
    \includegraphics[width=0.48\textwidth]{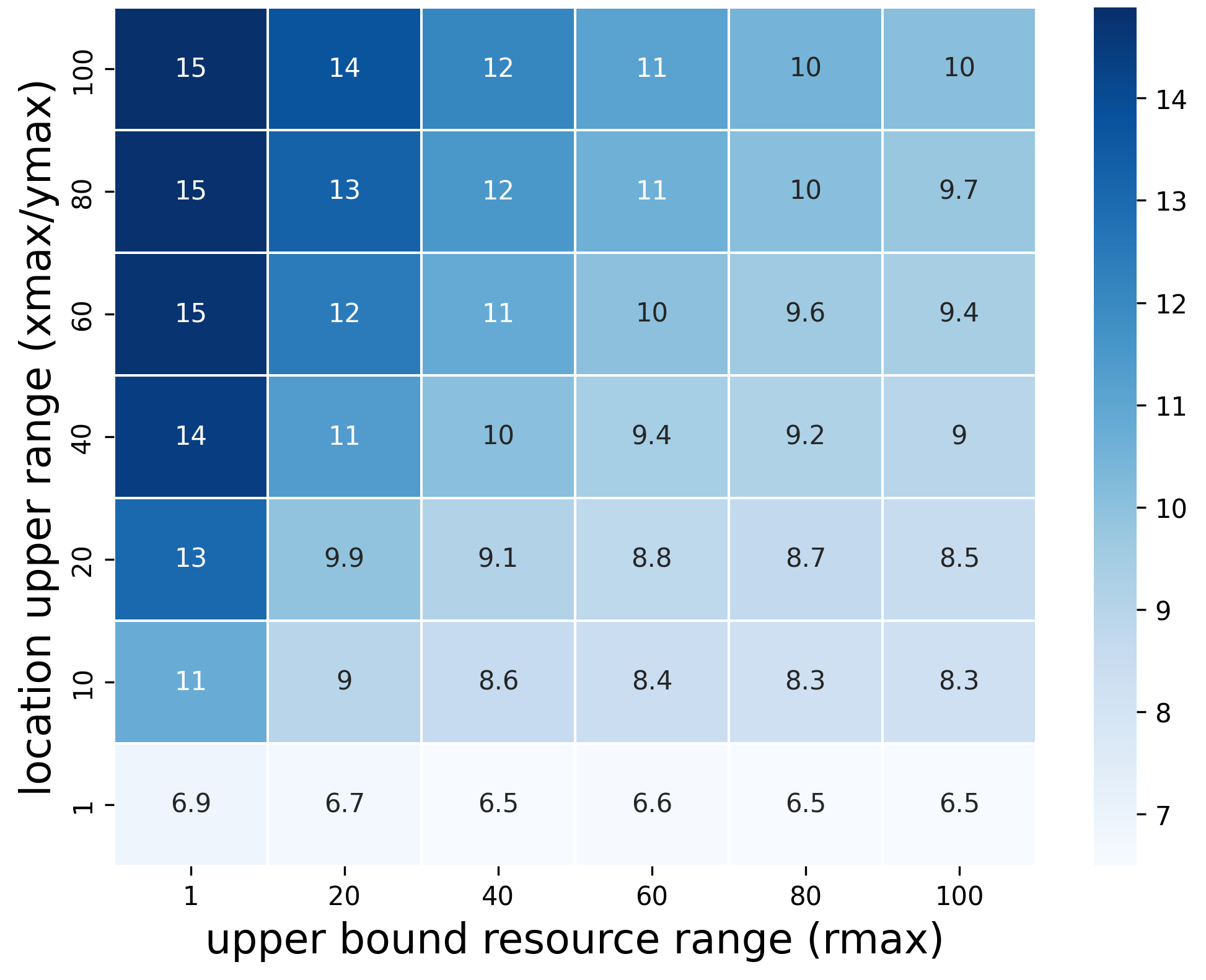}
    \caption{Number of groups across 1000 simulations/games for each setting (Grid).}
    \label{fig:sub1}
\end{figure}

\begin{figure}[ht]
    \centering
    \includegraphics[width=0.48\textwidth]{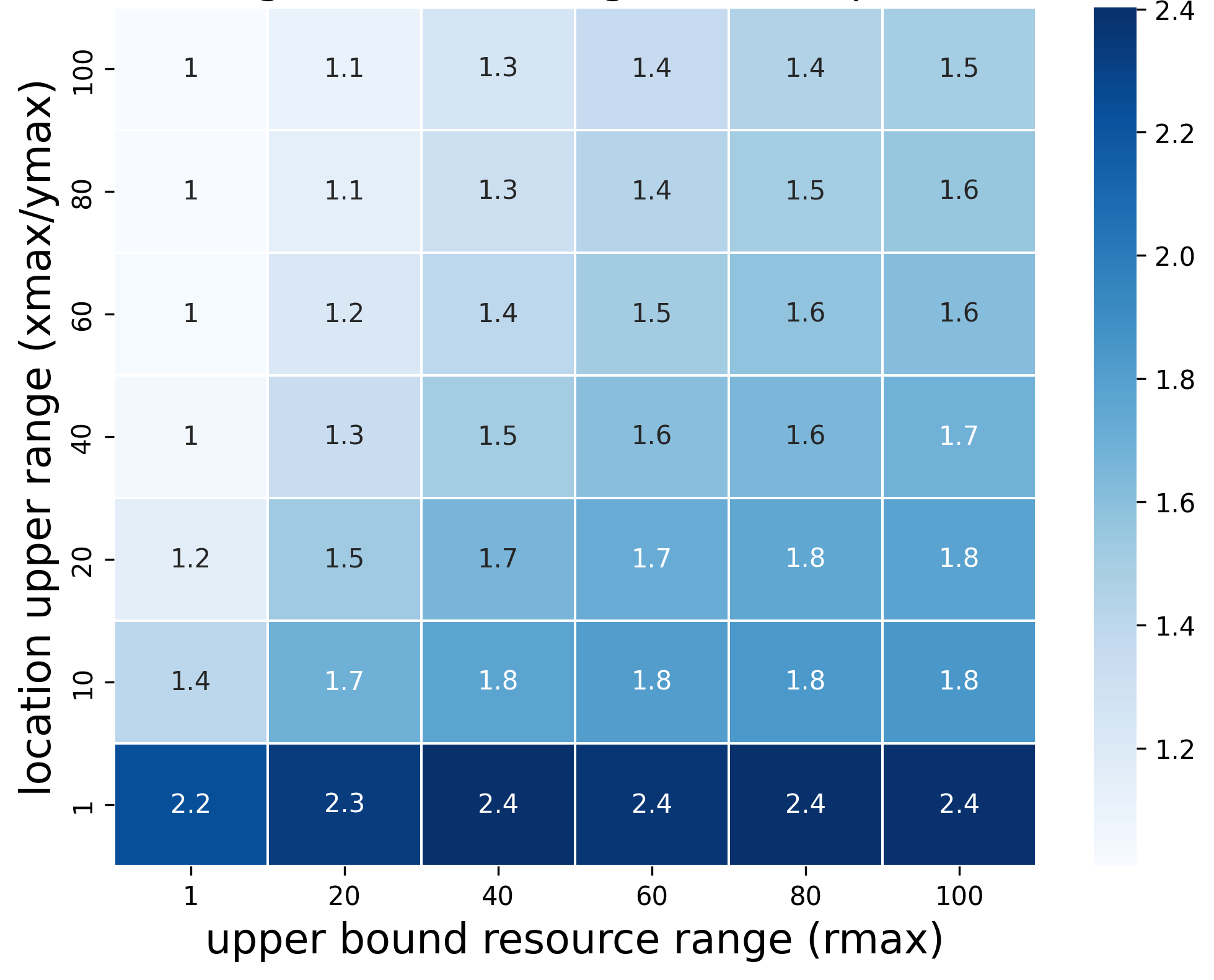}
    \caption{The average number of organizations in the stable groups of each game is calculated, and then this average is computed across 1000 games.}
    \label{fig:sub2}
\end{figure}

\begin{figure}[ht]
    \centering
    \includegraphics[width=0.48\textwidth]{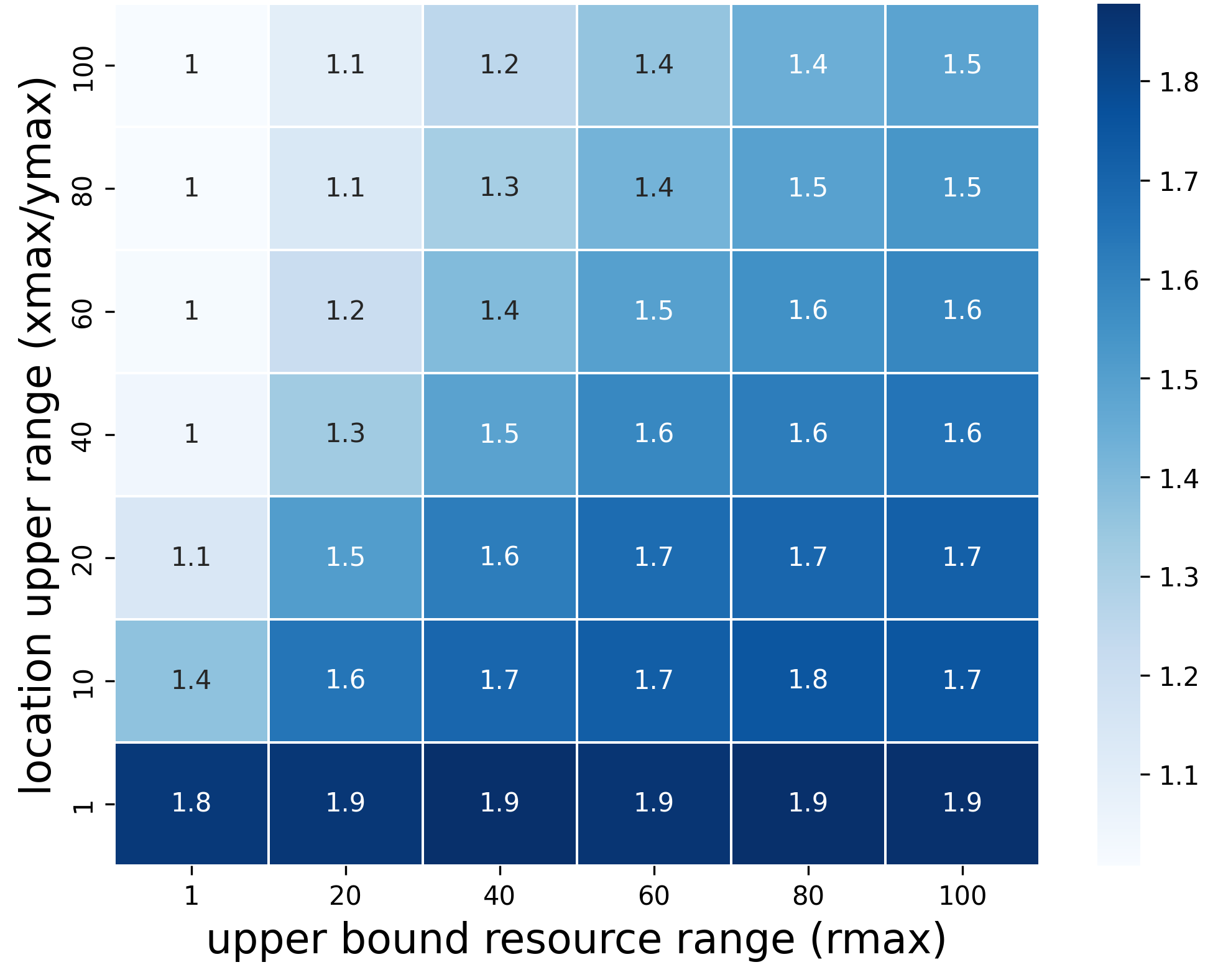}
    \caption{The average number of sectors in the stable groups of each game is calculated, and then this average is computed across 1000 games.}
    \label{fig:sub3}
\end{figure}

Some key results we draw from the simulations:
\begin{enumerate}
    \item As illustrated by the rows in Figure~\ref{fig:sub1} and Figure~\ref{fig:sub2}, an increase in the resource range leads to the formation of larger groups, with more organizations included in each group. This results in a decrease in the total number of groups formed across all organizations. This can be explained as follows: A larger resource range increases the likelihood of organizations' resource dependence because some organizations have many sector-specific resources, and some organizations have few sector-specific resources. This inter-organizational resource dependence, in turn, promotes the formation of larger groups (i.e., groups with more organizations and therefore more interorganizational dependencies) as one moves from the leftmost column across the figure to the rightmost column. 
    \item As illustrated by the columns in Figure~\ref{fig:sub1} and Figure~\ref{fig:sub2}, an increase in the location range leads to the formation of smaller groups, with fewer organizations or single organization included in each group (i.e., when the average number of organizations in a group is 1 in Figure~\ref{fig:sub2}; this is presently categorized as a singleton group). Note that when there is only one sector, no CSP formation occurs. In other words, CSPs may not form at all under certain conditions. This results in an increase in the total number of groups (including singletons) formed across all organizations. This can be explained as follows: a larger location range increases the distances among organizations, which in turn discourages the formation of groups and CSPs due to the higher spatial cost. An extreme case occurs when all organizations have the same resource quantity (i.e., the first column of Figure~\ref{fig:sub2}), in which no CSPs are formed (i.e., every organization chooses to be a singleton group) when the location range is sufficiently large (i.e., the top-right corner of Figure~\ref{fig:sub3}). At the other extreme, when all organizations have highly variable resource amounts (i.e., the last column of Figure~\ref{fig:sub2}), CSPs are always formed, even when the location range is large. This is because resource dependencies increase the need for inter-organizational partners, regardless of distance. 
    \item There is also a consistent pattern between Figures~\ref{fig:sub2} and~\ref{fig:sub3}: as the number of organizations in groups increases, the number of sectors within those groups also increases. Moreover, larger (i.e., more organizations) and more diverse (i.e., more sectors) groups are likely to form when organizations are located close to each other. When the location range is lowest (i.e., the bottom row of Figures~\ref{fig:sub2} and~\ref{fig:sub3}), group size and diversity are relatively consistent across resource ranges.
\end{enumerate}

Our agent-based model also enables us to observe how group formation unfolds over time under different conditions. The figures below illustrate a representative example of this process under specific settings. For clarity, we present six key iterations, beginning with individuals in singleton groups and ending with the convergence to a stable group structure. A full visualization of the dynamic process will be made available in the final version of the paper.

\begin{exmp}\label{eg:dynamic}
In this toy example, we choose 3 organizations in each of the three categories, which makes it 9 organizations in total, i.e. $\mathcal{N}=\{1,2,...,9\}$ and $\mathbf{t} = \{c_1, c_1, c_1, c_2, c_2, c_2, c_3, c_3, c_3\}$. Their individual resources in each of the three categories are 
$\mathbf{r}_{c_1} = \{10, 1, 9, 0,0,0,0,0,0\}$, 
$\mathbf{r}_{c_2} = \{0,0,0,5,2,20,0,0\}$, 
and $\mathbf{r}_{c_3} = \{0,0,0,0,0,0, 10, 8, 4\}$. In other words, the first individual's resource is $r_1 = (\mathbf{r}_{c_1,1}, \mathbf{r}_{c_2,1}, \mathbf{r}_{c_3,1}) = (10, 0, 0)$. The location of the individuals is $[\mathbf{x}, \mathbf{y} ] = [[1,2], [2,1], [0,0], [3,1], [1,1], [2,0], [2,3], [1,0], [3,2]]$.

After running the agent-based model simulation under this setting, the system converges in 21 iterations. Below, we highlight 6 key iterations where significant changes or new group formations occur. As shown in the snapshots, the formation process begins with isolated individuals and gradually leads to the emergence of stable cross-sector partnerships. Groups form by balancing spatial proximity and resource complementarity, with categorically (or sectorally, in the case of CSPs) heterogeneous coalitions becoming increasingly prominent as the system approaches convergence. Specifically, the process begins with 9 individuals in singleton groups. A 2-agent group forms first, followed by the formation of a 3-agent group. Subsequently, two separate 2-agent groups emerge. Some groups dissolve during the process, and the system eventually converges to a stable configuration consisting of one 3-agent group and three 2-agent groups.

\begin{figure}[ht]
    \centering
    \subfloat[initialization]{%
        \includegraphics[width=0.48\linewidth]{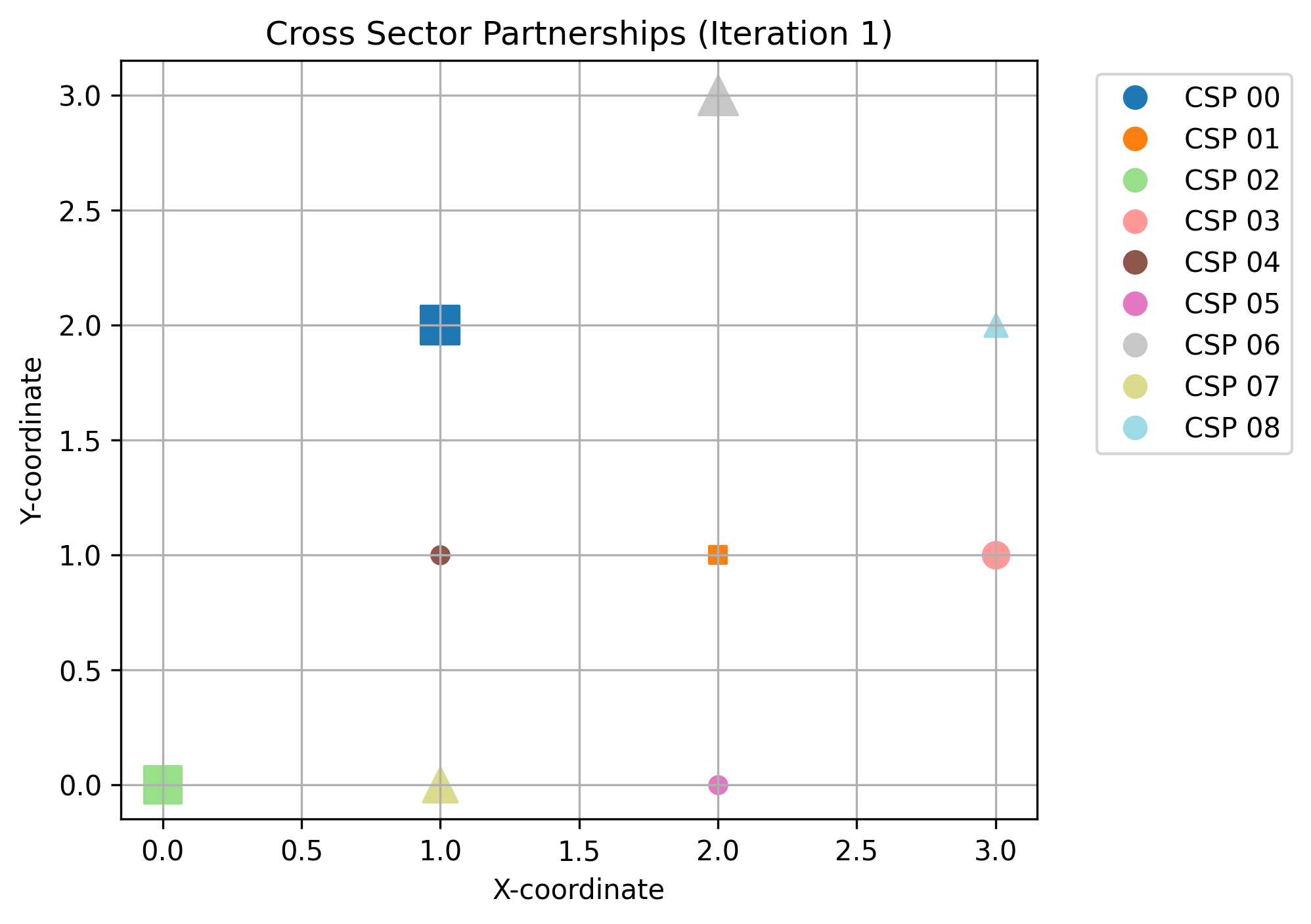}
        \label{fig:iter1}
    }
    \hfill
    \subfloat[iteration 1]{%
        \includegraphics[width=0.48\linewidth]{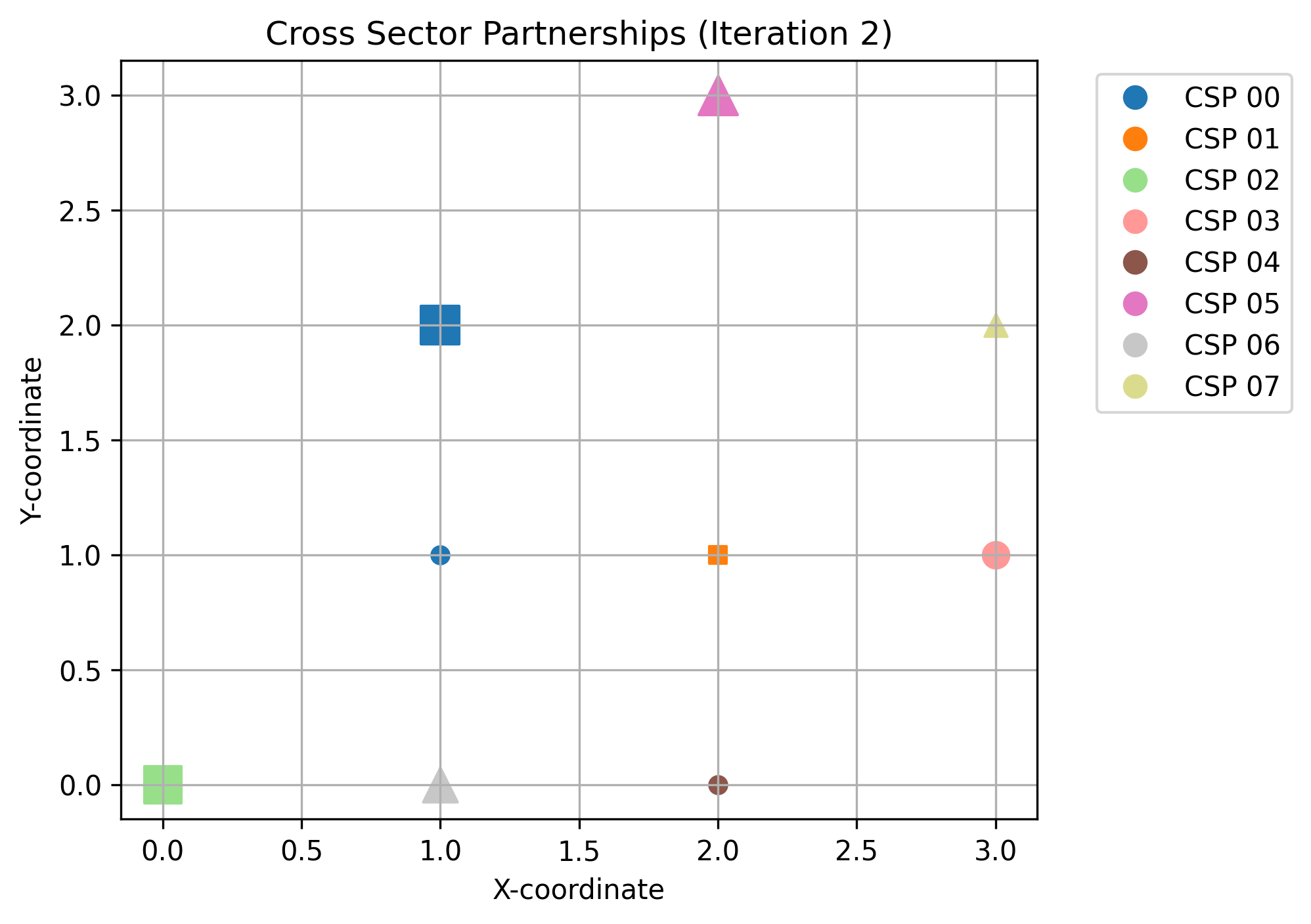}
        \label{fig:iter2}
    }
    \hfill
    \subfloat[iteration 2]{%
        \includegraphics[width=0.48\linewidth]{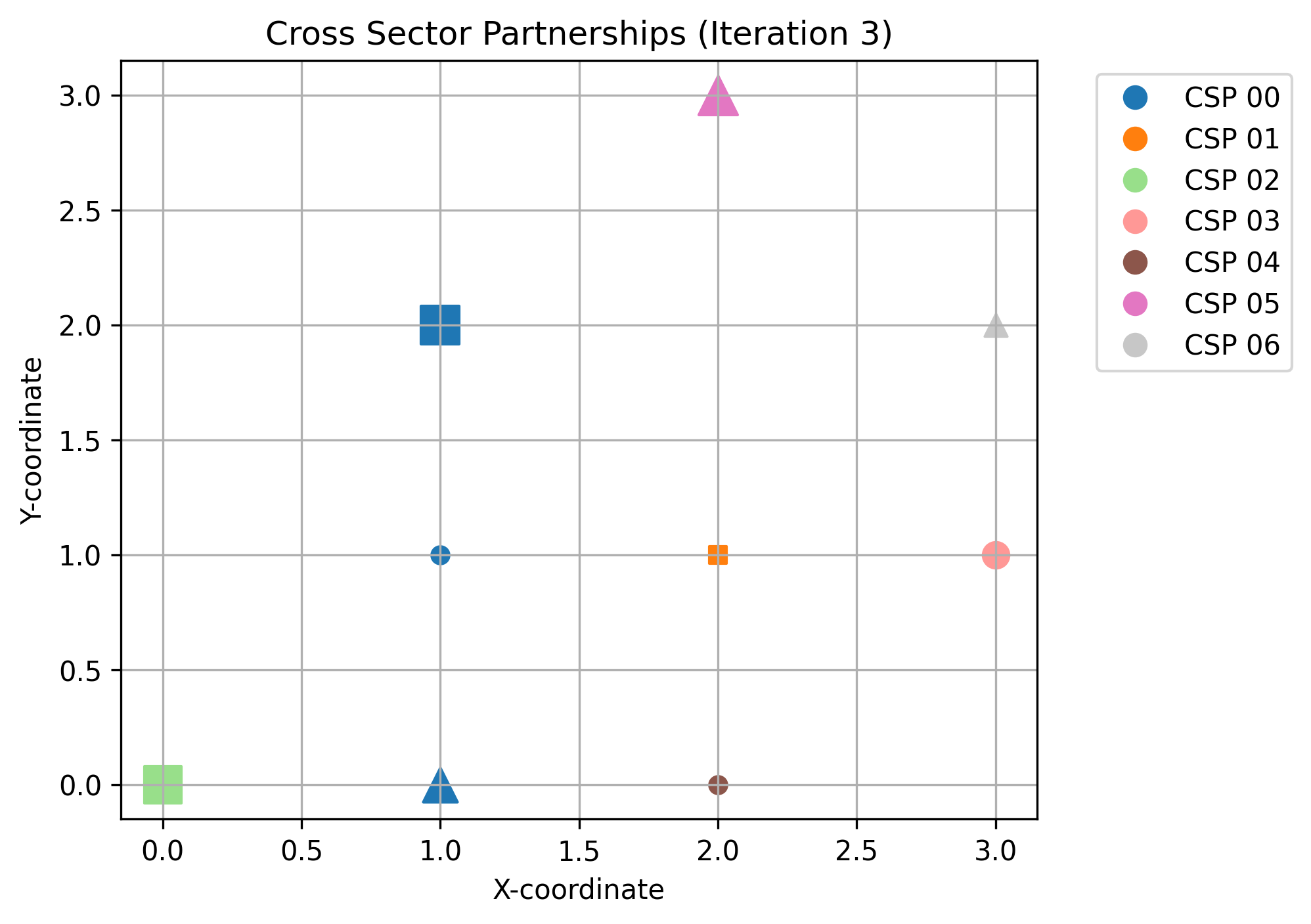}
        \label{fig:iter3}
    }
    \hfill
    \subfloat[iteration 8]{%
        \includegraphics[width=0.48\linewidth]{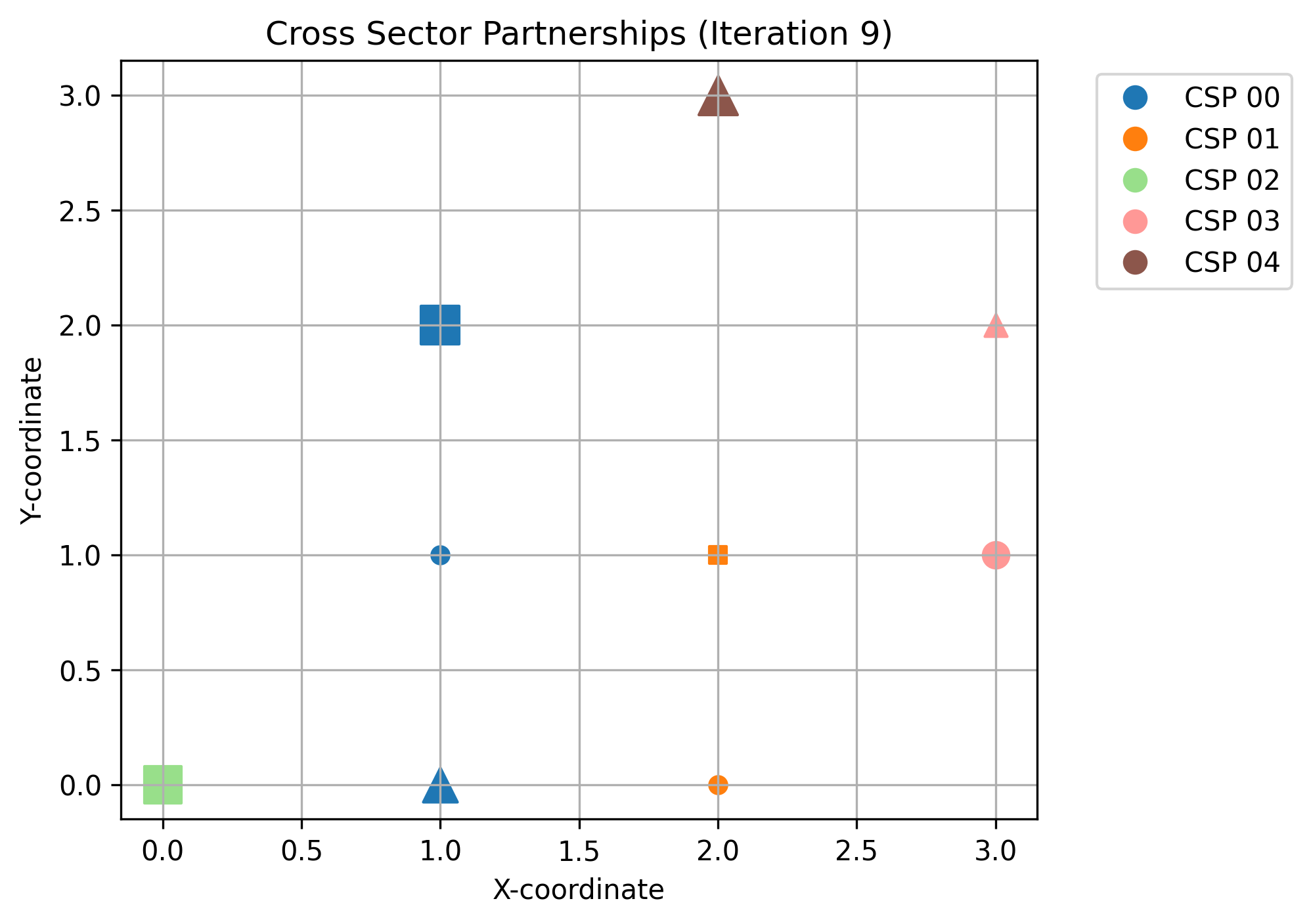}
        \label{fig:iter4}
    }
    \hfill
    \subfloat[iteration 11]{%
        \includegraphics[width=0.48\linewidth]{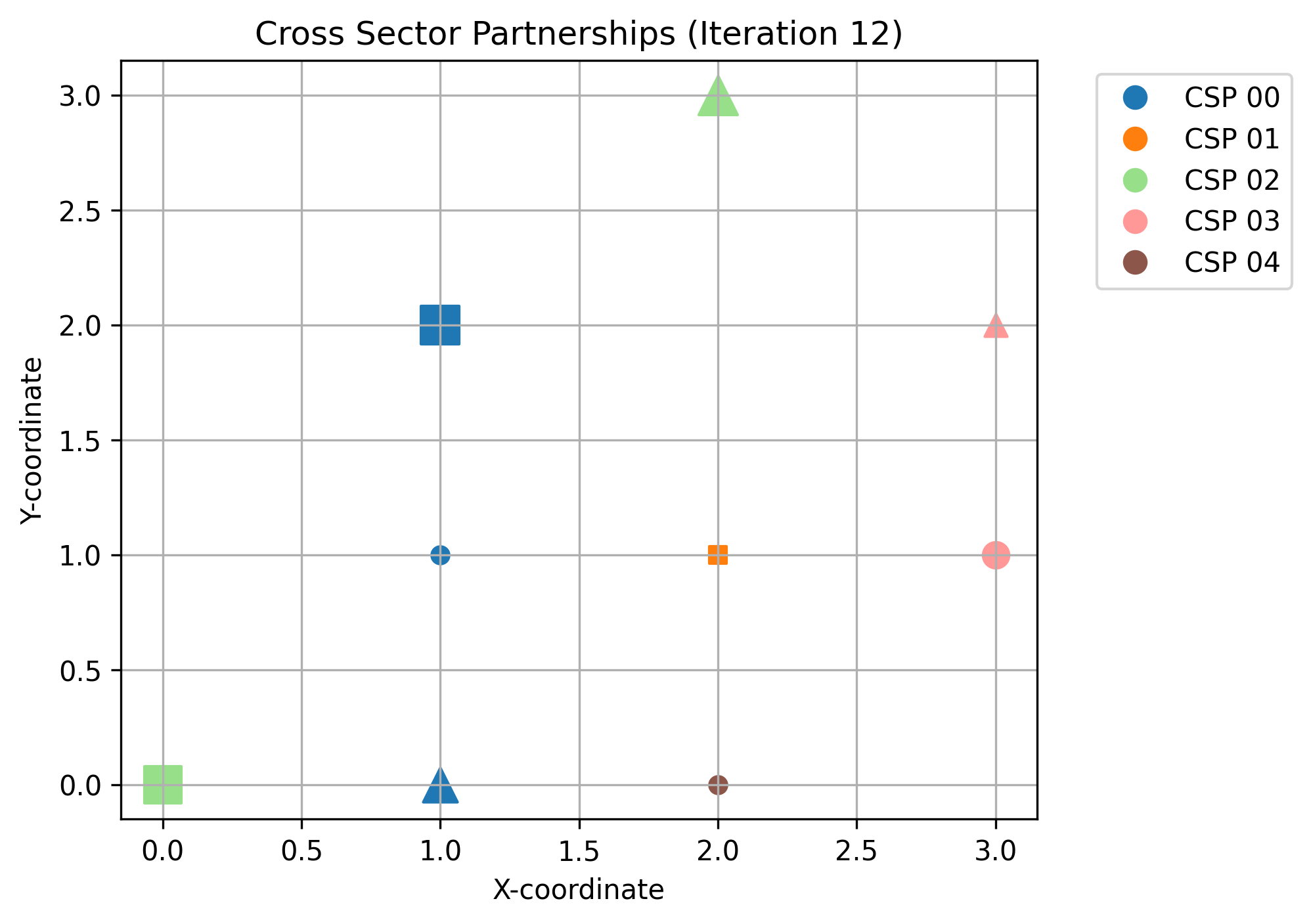}
        \label{fig:iter5}
    }
    \hfill
    \subfloat[iteration 21]{%
        \includegraphics[width=0.48\linewidth]{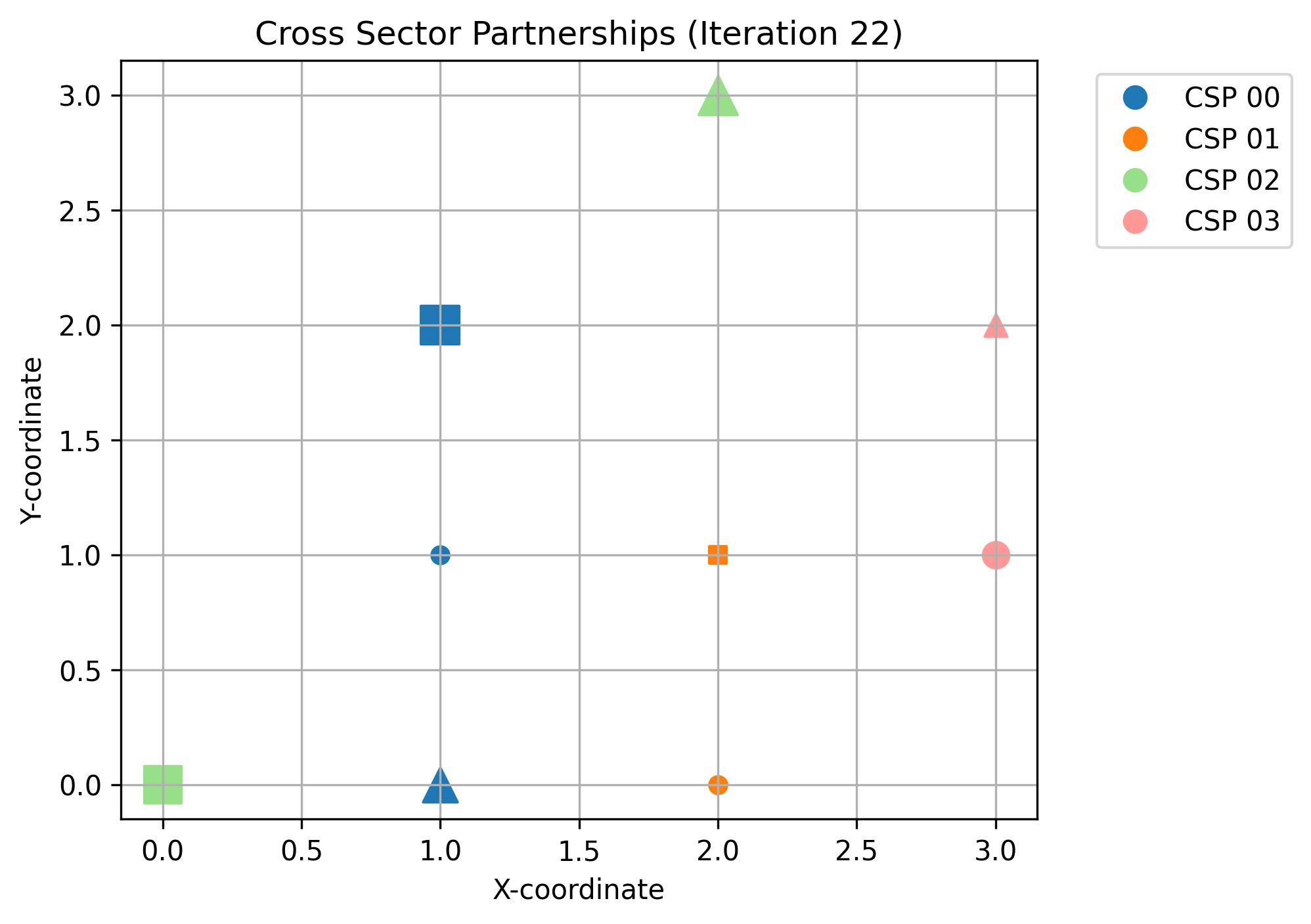}
        \label{fig:iter6}
    }
    \caption{Selected snapshots from the agent-based model simulation, illustrating six key iterations out of the 21 total iterations for convergence at an equilibrium. These iterations highlight major changes in group composition and the emergence of stable cross-sector partnerships. The shapes represent the categories: square for $c_1$, circle for $c_2$, and triangle for $c_3$. The size of each shape reflects the value of individual resources.}
    \label{fig:overall}
\end{figure}
\end{exmp}
\section{Conclusion}\label{sec:conclusion} 
In this paper, we developed a game-theoretic model, complemented by an agent-based model, to examine how resource dependence and spatial cohesion interact in the formation of groups such as cross-sector partnerships (CSPs). Essentially, this is an optimization problem where individuals try to maximize their utility by forming groups based on their rational anticipation of others’ strategies. Our analysis identifies resource pooling, resource heterogeneity, and spatial cohesion as fundamental drivers of group formation. Specifically, when individual resources are limited, groups form only among spatially proximate actors, as the cost of coordination outweighs the benefits of resource complementarity. In contrast, when individual resources are abundant, groups form even across greater distances, as the need for critical resources overrides geographic constraints. Furthermore, resource-heterogeneous groups tend to emerge more readily when organizations are located near one another, allowing them to capitalize on pooling and complementarity.

The agent-based model also allows us to trace the dynamic process of group formation across varying initial conditions. Our results show that increasing the range of available resources leads to larger, more inclusive groups but reduces the total number of distinct groups, as organizations rely more on shared resources to overcome spatial separation. In contrast, expanding the geographic range results in smaller or fragmented groups, or no group formation at all, due to heightened spatial costs. As group size increases, so does sectoral heterogeneity, reflecting the interplay between geography and resource availability in shaping group composition.

These findings underscore that successful group formation is driven not only by internal incentives like resource complementarity, but also by external constraints such as spatial proximity and organizational heterogeneity. For CSPs and similar collaborations to form and thrive, organizations must be both willing and able to engage, decisions shaped by their resource needs and spatial context. Understanding these mechanisms can help guide strategies to promote effective cross-sector collaborations, especially in settings where both access to resources and geographic constraints are critical to addressing complex societal or engineering challenges.
\nocite{*}
\bibliographystyle{unsrt}
\bibliography{myreference}

% \appendix

\end{document}